\documentclass[letterpaper, 10 pt, conference]{ieeeconf}  

\IEEEoverridecommandlockouts                              

\usepackage{graphicx} 
\usepackage{amsmath} 
\usepackage{amssymb}  

\usepackage{bm}
\usepackage{mathtools}
\usepackage{tabularx}
\usepackage{booktabs}
\usepackage{multirow}
\usepackage{units}
\usepackage{cite}

\newcommand{\diag}{\mathrm{diag}}

\newtheorem{theorem}{Theorem}
\newtheorem{assumption}{Assumption}
\newtheorem{remark}{Remark}
\newtheorem{definition}{Definition}
\newtheorem{proposition}{Proposition}

\def\ct-#1{ct\nobreakdash-#1}
\def\RPI-#1-#2{RPI\nobreakdash-#1\nobreakdash-#2}
\def\SL-#1{SL\nobreakdash-#1}
\def\tube-#1{tube\nobreakdash-#1}

\usepackage{color}

\newcolumntype{C}{>{\centering\arraybackslash}X}

\title{\LARGE \bf
System Level Disturbance Reachable Sets\\and their Application to Tube-based MPC
}

\author{Jerome Sieber$^{1}$, Andrea Zanelli$^{1}$, Samir Bennani$^{2}$, and Melanie N. Zeilinger$^{1}$
\thanks{*This work was supported by the European Space Agency (ESA) under NPI 621-2018 and the Swiss Space Center (SSC).%
}
\thanks{$^{1}$J. Sieber, A. Zanelli, and M. N. Zeilinger are members of the Institute for Dynamic Systems and Control (IDSC), ETH Zurich, 8092 Zurich, Switzerland
        {\tt\small \{jsieber,zanellia,mzeilinger\}@ethz.ch}%
}
\thanks{$^{2}$S. Bennani is a member of ESA-ESTEC, Noordwijk 2201 AZ, The Netherlands
        {\tt\small samir.bennani@esa.int}%
}
}

\begin{document}

\maketitle
\thispagestyle{empty}
\pagestyle{empty}

\begin{abstract}
Tube-based model predictive control~(MPC) methods leverage tubes to bound deviations from a nominal trajectory due to uncertainties in order to ensure constraint satisfaction. This paper presents a novel \tube-based MPC formulation based on system level disturbance reachable sets~(\SL-DRS), which leverage the affine system level parameterization~(SLP). We show that imposing a finite impulse response~(FIR) constraint on the affine SLP guarantees containment of all future deviations in a finite sequence of \SL-DRS. This allows us to formulate a system level \tube-MPC~(SLTMPC) method using the \SL-DRS as tubes, which enables concurrent optimization of the nominal trajectory and the tubes, while using a positively invariant terminal set. Finally, we show that the \SL-DRS tubes can also be computed offline.
\end{abstract}


\section{INTRODUCTION}
Tube-based model predictive control~(MPC) methods have become the principal control techniques for robust control of constrained linear systems with bounded additive disturbances and bounded parametric model uncertainties. The main concept is to separate the system behavior into nominal and error dynamics, which describe the desired nominal system behavior neglecting the uncertainties and the deviations from this nominal behavior due to the uncertainties, respectively. Constraint satisfaction in spite of the present uncertainties is achieved by tightening the constraints using a worst case bound on the error trajectories, also referred to as tubes. The earliest proposed tube parameterizations use a constant offline-computed tube controller and define tubes with constant cross-sections~\cite{Langson2004,Mayne2005} or growing cross-sections~\cite{Chisci2001,Fleming2014}. These parameterizations were then extended to homothetic and elastic tubes~\cite{Rakovic2012b,Rakovic2016}, whose cross-sections have a fixed shape but can be dilated and translated. The most general tube parameterization is introduced in~\cite{Rakovic2012a}, which uses tubes that are parameterized by the convex hulls of the predicted state and input trajectories. Another class of parameterizations implicitly parameterize the tubes through the optimization of feedback gains, which are either defined as affine in the predicted states~\cite{Lofberg2003} or as affine in the past disturbances~\cite{Goulart2006,Sieber2021}.

In this paper, we derive a novel tube parameterization based on the affine system level parameterization~(SLP)~\cite{Sieber2021}, which is an extended version of the SLP that was introduced as part of the system level synthesis framework~\cite{Anderson2019}. The affine SLP offers the advantage that it enables optimization over closed-loop trajectories instead of controller gains. We use this to parameterize a sequence of tubes in terms of the closed-loop error trajectories, enabling concurrent optimization of the nominal trajectory and the error trajectories and therefore the tubes. Subsequently, we use these SLP-based tubes to formulate a \tube-based MPC formulation. The first SLP-based MPC formulations for additive disturbances and parametric model uncertainties were introduced in~\cite{Anderson2019} and~\cite{Chen2020}, respectively. A more conservative variant of the latter approach was already introduced for the constrained linear quadratic regulator (LQR) in~\cite{Dean2018}. Due to the separable structure of the SLP, SLP-based MPC formulations have been mainly proposed for distributed systems~\cite{Chen2019,Alonso2019,Alonso2020,Li2021}. More recently, a direct correspondence between \tube-based MPC and SLP-based MPC formulations was established in~\cite{Sieber2021}.

\textit{Contributions:} We consider linear time-invariant dynamical systems with additive uncertainties. For this class of systems, we propose a variant of disturbance reachable sets~(DRS) derived from the affine SLP, which we call system level DRS~(\SL-DRS). Subsequently, we show that imposing a finite impulse response~(FIR) constraint on the system responses defining the \SL-DRS, results in the containment of all possible infinite horizon error trajectories in a finite sequence of \SL-DRS, which allows us to formulate a \tube-based MPC method using the \SL-DRS as tubes. The proposed method enables concurrent optimization of the nominal trajectory and the tubes, while using a positively invariant terminal set. Finally, we show that these tubes can also be computed offline.

\textit{Related Work:} In the context of min-max MPC,~\cite{Scokaert1998} introduced the idea of defining invariant sets through linear FIR systems, i.e., systems which converge to the origin in a finite number of time steps. This is achieved by computing a constant feedback controller which renders the closed-loop system FIR. In contrast, the proposed \SL-DRS use an implicitly defined set of controllers instead of a single one and explicitly enforce a specific FIR horizon instead of only requiring a finite one. A \tube-based MPC method parameterizing the tubes as a linear combination of multiple precomputed tubes was introduced in~\cite{Kogel2020}. Instead of fixing the tubes a~priori, the proposed \SL-DRS directly use optimized controller gains and guarantee that all infinite horizon error trajectories lie in a finite sequence of \SL-DRS. This work is closest related to~\cite{Sieber2021}, which introduced the system level \tube-MPC~(SLTMPC) method that forms the basis of the \tube-based MPC method presented in this paper. However, the SLTMPC method makes no use of an explicit tube formulation and requires a robust positively invariant terminal set, while the presented method uses the \SL-DRS and only requires a positively invariant terminal set.

The remainder of the paper is organized as follows: Section~\ref{sec:preliminaries} introduces the notation, basic definitions, and the concepts of \tube-based MPC and the affine SLP. We propose the FIR-constrained \SL-DRS in Section~\ref{sec:SL-DRS}, before formulating two \tube-based MPC methods using the \SL-DRS as tubes in Section~\ref{sec:SLTMPC}. Finally, Section~\ref{sec:numerical_section} presents numerical experiments and Section~\ref{sec:conclusions} concludes the paper.

\section{PRELIMINARIES}\label{sec:preliminaries}
\subsection{Notation and Basic Definitions}
In this paper, we denote vectors with lower case letters, e.g.,~$a$, and matrices with upper case letters, e.g.,~$A$. With bold letters we indicate stacked column vectors, e.g.,~$\mathbf{a}$, and block matrices, e.g.,~$\mathbf{A}$, whose elements we index with superscripts, e.g.,~$a^i, \ A^{i,j}$. For two vectors~$x, \, y$, we use the shorthand~$(x, y)$ to denote $[x^\top, y^\top]^\top$ and we frequently consider block-lower-triangular Toeplitz matrices, which are defined by the structure
\begin{equation*}
\mathbf{A} = \begin{bmatrix} A^{0} & & & \\ A^{1} & A^{0} & & \\ \vdots & \ddots & \ddots & \\ A^{N} & \dots & A^1 & A^{0} \end{bmatrix}\!.
\end{equation*}
Additionally, we use~$\diag(A)^i$ to denote the block-diagonal matrix, whose diagonal consists of $i$-times the $A$ matrix. We distinguish between the states of a dynamical system~$x(t)$ and the states predicted by an MPC algorithm~$x_i$. We use~$\oplus$ and~$\ominus$ to denote Minkowski set addition and Pontryagin set subtraction, respectively, which are both defined in~\cite[Definition~3.10]{Rawlings2009}, and for a sequence of sets~$\mathcal{A}_0, \dots, \mathcal{A}_j$, we use the convention that~$\bigoplus_{i=0}^{-1} \mathcal{A}_i = \{0\}$, which is the set containing only the zero vector. Below, we state the definitions of set invariance~\cite{Blanchini99}, disturbance reachable sets~(DRS)~\cite{Kouvaritakis2016}, and input-to-state stability~(ISS)~\cite{Limon2009}, which we will use throughout the paper.
\begin{definition}[Positively invariant set]
A set $\mathcal{S} \subseteq \mathbb{R}^n$ is a positively invariant (PI) set for the system $x^+ = Ax, \, A \in \mathbb{R}^{n \times n}$, if $x^+ \in \mathcal{S}$ for all $x \in \mathcal{S}$.
\end{definition}
\begin{definition}[Robust positively invariant set]
A set $\mathcal{S} \subseteq \mathbb{R}^n$ is a robust positively invariant (RPI) set for the system $x^+ = Ax + w, \, A \in \mathbb{R}^{n \times n}$, if $x^+ \in \mathcal{S}$ for all $x \in \mathcal{S}$, $w \in \mathcal{W} \subseteq \mathbb{R}^n$.
\end{definition}
\begin{definition}[Disturbance reachable sets]
For the system $x^+ = Ax + w, \, A \in \mathbb{R}^{n \times n},\, w \in \mathcal{W} \subseteq \mathbb{R}^n$, where~$\mathcal{W}$ is a compact set containing the origin in its interior, the disturbance reachable sets~(DRS) are defined as~$\mathcal{S}_i \coloneqq \bigoplus_{j=0}^{i-1}A^j\mathcal{W}, \, i \geq 0$.
\end{definition}
\begin{definition}[$\mathcal{K}$- and $\mathcal{K}_\infty$-functions]\label{def:DRS}
A function~$\alpha: \mathbb{R}_{\geq 0} \to \mathbb{R}_{\geq 0}$ is called a $\mathcal{K}$-function, if it is continuous, strictly increasing, and~$\alpha(0) = 0$. If the $\mathcal{K}$-function~$\alpha$ also fulfills~$\alpha(s) \to \infty$ as~$s \to \infty$, then we call~$\alpha$ a $\mathcal{K}_\infty$-function.
\end{definition}
\begin{definition}[ISS Lyapunov function]\label{def:ISS-Lyap}
Consider an RPI set~$\mathcal{S} \subseteq \mathbb{R}^n$ with the origin in its interior. A function $V: \mathbb{R}^n \to \mathbb{R}_{\geq 0}$ is called an ISS Lyapunov function in~$\mathcal{S}$ for the system $x^+ = Ax + w, \, A \in \mathbb{R}^{n \times n}$, with respect to~${w \in \mathcal{W} \subseteq \mathbb{R}^n}$, if there exist suitable $\mathcal{K}_\infty$-functions~$\alpha_1, \alpha_2, \alpha_3$ and a $\mathcal{K}$-function~$\gamma$, such that:
\begin{subequations}
\begin{align}
&\hspace*{-0.08cm}\alpha_1(\|x\|) \leq V(x) \leq \alpha_2(\|x\|), \ \forall x \in \mathcal{S}, \label{ISS-L:bounds}\\
&\hspace*{-0.08cm}V(x^+) \!-\! V(x) \leq \!-\alpha_3(\|x\|) \!+\! \gamma(\|w\|), \, \forall x \!\in\! \mathcal{S},\, w \!\in\! \mathcal{W}, \label{ISS-L:decrease}
\end{align}
\end{subequations}
where~$\| \cdot \|$ is any vector norm on~$\mathbb{R}^n$.
\end{definition}

\subsection{Tube-based MPC}
We consider linear time-invariant (LTI) dynamical systems with additive disturbances of the form
\begin{equation}\label{eq:dynamics}
x(t\!+\!1) = A x(t) + B u(t) + w(t),
\end{equation}
with $A \in \mathbb{R}^{n \times n}$, $B \in \mathbb{R}^{n \times m}$, and $w(t) \in \mathcal{W}$, where $\mathcal{W} = \{ w \!\mid\! H_w w \leq h_w \} \subseteq \mathbb{R}^n$ is a compact set containing the origin in its interior. The system is subject to compact polytopic state and input constraints
\begin{subequations}\label{eq:constraints}
\begin{align}
\mathcal{X} &= \{ x \in \mathbb{R}^n \mid H_{x} x \leq h_{x}, \, H_{x} \in \mathbb{R}^{n_x \times n}\}, \\
\mathcal{U} &= \{ u \in \mathbb{R}^m \mid H_{u} u \leq h_{u}, \, H_{u} \in \mathbb{R}^{n_u \times m}\},
\end{align}
\end{subequations}
containing the origin in their interior. To control system~\eqref{eq:dynamics}, we consider MPC with horizon length~$N$. Hence, we define $\tilde{\mathbf{x}} = (1, \mathbf{x})$ and $\tilde{\bm{\delta}} = (1, \bm{\delta})$, where $\mathbf{x} = ( x_0, x_1, \dots, x_N )$ and $\bm{\delta} = ( x_0, \mathbf{w} ) = ( x_0,w_0, w_1, \dots, w_{N-1} )$, as the affine state and disturbance trajectories, respectively, and $\mathbf{u} = ( u_0, u_1, \dots, u_N )$ as the input trajectory\footnote{Note that we include the input~$u_N$ here to simplify the notation for the SLP. However, this input will be ignored in the MPC formulations.}. The dynamics can then be compactly written in terms of these trajectories:
\begin{equation}\label{eq:stacked_dynamics}
\tilde{\mathbf{x}} = \mathcal{ZA}\tilde{\mathbf{x}} + \mathcal{ZB}\mathbf{u} + \tilde{\bm{\delta}},
\end{equation}
where $\mathcal{Z}$ is the down-shift operator and $\mathcal{A}$, $\mathcal{B}$ are the matrices
\begin{align*}
\mathcal{A} = \begin{bmatrix} 0 & \\  & \diag(A)^{N+1} \end{bmatrix},\quad
\mathcal{B} = \begin{bmatrix} 0 \\ \diag(B)^{N+1} \end{bmatrix}.
\end{align*}
In this paper, we focus on \tube-based MPC methods~\cite{Rawlings2009} for the uncertain system~\eqref{eq:dynamics} subject to~\eqref{eq:constraints}. These methods define the control policy as
\begin{align}
\mathbf{u}^\textup{tube} = \begin{bmatrix} \mathbf{v} - \mathbf{Kz} & \mathbf{K} \end{bmatrix}\tilde{\mathbf{x}} &=  \mathbf{K}\left(\mathbf{x}-\mathbf{z}\right) + \mathbf{v},
\end{align}
where $\mathbf{z} = ( z_0, \dots, z_{N} )$, $\mathbf{v} = ( v_0, \dots, v_{N} )$, and $\mathbf{K} = \diag(K)^{N+1}$ are the nominal states, nominal inputs, and tube controller, respectively. Due to the linearity of the system, dynamics~\eqref{eq:stacked_dynamics} can be split into nominal and error dynamics, i.e.,
\begin{align}
\mathbf{z} &= \mathcal{ZA}\mathbf{z} + \mathcal{ZB}\mathbf{v}, \label{eq:nominal-dynamics}
 \\
\mathbf{x} - \mathbf{z} &= \mathbf{e} = \left(\mathcal{ZA} + \mathcal{ZB}\mathbf{K}\right)\mathbf{e} + \bm{\delta}. \label{eq:error-dynamics}
\end{align}
This allows us to recast the \tube-based MPC problem in the nominal variables $\mathbf{z},\, \mathbf{v}$ instead of the system variables $\mathbf{x},\, \mathbf{u}$, while imposing tightened state and input constraints on the nominal variables via bounds on the error dynamics, i.e., the tubes. The resulting \tube-based MPC problem is given by
\begin{subequations}\label{MPC:generic}
	\begin{alignat}{2}
		\min_{\mathbf{z}, \mathbf{v}} \quad & l_f(z_N) + \sum_{i=0}^{N-1}l(z_i,v_i), \label{MPC:cost}\\
		\textrm{s.t. } \:\; & \mathbf{z} = \mathcal{ZA}\mathbf{z} + \mathcal{ZB}\mathbf{v} \\
		& z_i \in \mathcal{X} \ominus \mathcal{F}_i, && i=0, \dots N-1,\\
		& v_i \in \mathcal{U} \ominus K\mathcal{F}_i, && i=0, \dots N-1, \\
		& z_N \in \mathcal{Z}_f, \\
        & z_0 \in \mathcal{Z}_0
	\end{alignat}
\end{subequations}
where $l(\cdot,\cdot)$ and $l_f(\cdot)$ are suitable stage cost and terminal cost functions, respectively, and~$K$ is the tube controller. The terminal set~$\mathcal{Z}_f$, the sequence of tightening sets~$\mathcal{F}_i$, and the initial state constraint~$z_0 \in \mathcal{Z}_0$ depend on the specific \tube-based MPC method used. The two most widely used methods are constraint tightening MPC~(\ct-MPC)~\cite{Chisci2001} and \RPI-tube-MPC~\cite{Mayne2005}, which both fix the tube controller~$K$ before solving~\eqref{MPC:generic}, but differ in the definitions of~$\mathcal{Z}_f$, $\mathcal{F}_i$, and $\mathcal{Z}_0$.

\emph{Constraint tightening MPC~\cite{Chisci2001}} defines the tightening sets as the DRS for error system~\eqref{eq:error-dynamics} under tube controller~$K$, i.e.,~$\mathcal{F}_i = \bigoplus_{j=0}^{i-1}\left( A + BK\right)^j\mathcal{W}$. The terminal set is then selected as~$\mathcal{Z}_f = \mathcal{X}_f \ominus \mathcal{F}_N$, where~$\mathcal{X}_f$ is an RPI set under tube controller~$K$. By convention, the initial tightening set~$\mathcal{F}_0$ only contains the zero vector, therefore the initial state is not subject to tightening and the MPC problem~\eqref{MPC:generic} is initialized with the measured state, i.e.,~$\mathcal{Z}_0 = \{x(t)\}$.

\emph{RPI-tube-MPC~\cite{Mayne2005}} defines the tightening sets as~$\mathcal{F}_i = \Omega, \, \forall i = 0, \dots, N-1$, where $\Omega = \bigoplus_{j=0}^\infty\left( A + BK\right)^j\mathcal{W}$ is the minimal RPI set for error system~\eqref{eq:error-dynamics} under tube controller~$K$. The terminal set is required to lie in the tightened state constraints, i.e.,~$\mathcal{Z}_f \subseteq \mathcal{X} \ominus \Omega$, where~$\mathcal{Z}_f$ is defined as a PI set, rather than an RPI set, since the error state is guaranteed to lie in $\Omega$ for all time steps. In order to reduce conservativeness, \RPI-tube-MPC optimizes over the initial state~$z_0$, i.e., the initial state constraint becomes~$z_0 \in \mathcal{Z}_0 = \{x(t)\} \oplus \Omega$.

While we focus on these two formulations for comparison, there exist many other variants, see, e.g.,~\cite{Rawlings2009,Kouvaritakis2016} for a more detailed treatment of the above formulations and a broader overview. Note that the DRS converge to the minimal RPI set in the limit with respect to the Hausdorff distance. In Section~\ref{sec:SL-DRS}, we propose a modified version of the DRS based on the affine system level parameterization~(SLP)~\cite{Sieber2021}.

\subsection{Affine System Level Parameterization}
Instead of parameterizing the nominal and error dynamics~\eqref{eq:nominal-dynamics},~\eqref{eq:error-dynamics} in terms of the tube controller~$\mathbf{K}$, the affine SLP parameterizes~\eqref{eq:nominal-dynamics},~\eqref{eq:error-dynamics} in terms of the closed-loop system. To formulate the closed-loop dynamics, we define an affine linear time-varying state feedback controller $\mathbf{u} = \mathbf{K}\tilde{\mathbf{x}}$, with
\begin{equation*}
\mathbf{K} = \begin{bmatrix} v^0 & K^{0,0} & & & \\ v^1 & K^{1,1} & K^{1,0} &  &  \\ \vdots & \vdots & \vdots & \ddots & \\ v^N & K^{N,N} & K^{N,N-1} & \dots & K^{N,0} \end{bmatrix},
\end{equation*}
resulting in the closed-loop trajectory $$\tilde{\mathbf{x}} = \left(\mathcal{ZA} + \mathcal{ZB}\mathbf{K}\right)\tilde{\mathbf{x}} + \tilde{\bm{\delta}} = \left(I - \mathcal{ZA} - \mathcal{ZB}\mathbf{K}\right)^{-1}\tilde{\bm{\delta}}.$$ Then, the \emph{affine system responses} are denoted by~$\bm{\Phi}_\mathbf{x}$ and~$\bm{\Phi}_\mathbf{u}$ describing the linear maps~$\tilde{\mathbf{x}} = \bm{\Phi}_\mathbf{x}\tilde{\bm{\delta}}$ and $\mathbf{u} = \bm{\Phi}_\mathbf{u}\tilde{\bm{\delta}}$, respectively, with~$\bm{\Phi}_\mathbf{x}, \bm{\Phi}_\mathbf{u}$ defined as
\begin{subequations}\label{eq:A-SLP}
\begin{align}
\bm{\Phi}_\mathbf{x} &= \left(I - \mathcal{ZA} - \mathcal{ZB}\mathbf{K}\right)^{-1}\!, \\ 
\bm{\Phi}_\mathbf{u} &= \mathbf{K}\left(I - \mathcal{ZA} - \mathcal{ZB}\mathbf{K}\right)^{-1}\!.
\end{align}
\end{subequations}
The affine system responses~$\bm{\Phi}_\mathbf{x}$,\, $\bm{\Phi}_\mathbf{u}$ completely define the behavior of the closed-loop system with affine state feedback~$\mathbf{K}$ and admit the following block-lower-triangular structure 
\begin{equation}\label{eq:SLP_structure}
\bm{\Phi}_\mathbf{x} = \begin{bmatrix} 1 & 0 \\ \bm{\phi}_\mathbf{z} & \bm{\Phi}_\mathbf{e}\end{bmatrix}\!, \quad \bm{\Phi}_\mathbf{u} = \begin{bmatrix} \bm{\phi}_\mathbf{v} & \bm{\Phi}_\mathbf{k}\end{bmatrix}\!,
\end{equation}
where~$\bm{\phi}_\mathbf{z} \in \mathbb{R}^{(N+1)n}, \ \bm{\phi}_\mathbf{v} \in \mathbb{R}^{(N+1)m}$ denote the nominal system responses and~$\bm{\Phi}_\mathbf{e} \in \mathbb{R}^{(N+1)n\times (N+1)n}, \ \bm{\Phi}_\mathbf{k} \in \mathbb{R}^{(N+1)m\times (N+1)n}$ are block-lower-triangular matrices denoting the error system responses. Then, we can reparameterize~\eqref{eq:nominal-dynamics},~\eqref{eq:error-dynamics} in terms of the system responses~$\bm{\Phi}_\mathbf{x}, \bm{\Phi}_\mathbf{u}$ instead of the affine state feedback gain $\mathbf{K}$ by exploiting the following theorem.
\begin{theorem}{(Theorem 3 in \cite{Sieber2021})}\label{theorem:A-SLP}
Consider the system dynamics~\eqref{eq:stacked_dynamics} over a horizon $N$ with block-lower-triangular state feedback law~$\mathbf{K}$ defining the control action as $\mathbf{u} = \mathbf{K}\tilde{\mathbf{x}}$. The following statements are true:
\begin{enumerate}
	\item the affine subspace defined by
	\begin{equation}\label{affine_SLS:affine-halfspace}\begin{bmatrix} I - \mathcal{Z}\mathcal{A} & -\mathcal{Z}\mathcal{B}\end{bmatrix} \begin{bmatrix} \bm{\Phi}_\mathbf{x} \\ \bm{\Phi}_\mathbf{u} \end{bmatrix} = I, \end{equation} parameterizes all possible affine system responses $\bm{\Phi}_\mathbf{x}$,~$\bm{\Phi}_\mathbf{u}$ with structure~\eqref{eq:SLP_structure},
	\item for any block-lower-triangular matrices $\bm{\Phi}_\mathbf{x}$, $\bm{\Phi}_\mathbf{u}$ satisfying~\eqref{affine_SLS:affine-halfspace}, the controller $\mathbf{K} = \bm{\Phi}_\mathbf{u} \bm{\Phi}_\mathbf{x}^{-1}$ achieves the desired affine system responses.
\end{enumerate}
\end{theorem}
In this paper, we restrict~$\bm{\Phi}_\mathbf{e}, \bm{\Phi}_\mathbf{k}$ to block-lower-triangular Toeplitz matrices for simplicity, i.e., the affine system responses take the form
\begin{equation*}\label{eq:mb_structure}
\bm{\Phi}_\mathbf{x} \!=\!\!\begin{bmatrix} \hspace{-.3em}1 & & & \\ \hspace{-.1em}\phi_z^0 & \hspace{-.1em}\Phi_e^0 & & \\ \hspace{-.1em}\vdots & \hspace{-.1em}\vdots & \hspace{-.3em}\ddots & \\ \hspace{-.1em}\phi_z^N & \hspace{-.1em}\Phi_e^N & \hspace{-.3em}\dots & \hspace{-.3em}\Phi_e^0 \end{bmatrix}\!\!, \;
\bm{\Phi}_\mathbf{u} \!=\!\!\begin{bmatrix} \hspace{-.1em}\phi_v^0 & \hspace{-.1em}\Phi_k^0 &  &  \\ \hspace{-.1em}\vdots & \hspace{-.1em}\vdots & \hspace{-.3em}\ddots &  \\ \hspace{-.1em}\phi_v^N & \hspace{-.1em}\Phi_k^N & \hspace{-.3em}\dots & \hspace{-.3em}\Phi_k^0 \end{bmatrix}\!\!.
\end{equation*}
However, all presented results can be extended to~$\bm{\Phi}_\mathbf{e}, \bm{\Phi}_\mathbf{k}$ with a general block-lower-triangular structure.

\section{SYSTEM LEVEL DISTURBANCE REACHABLE SETS}\label{sec:SL-DRS}
In this section, we use the system level perspective on the nominal and error dynamics to formulate a variant of the DRS called system level DRS~(\SL-DRS), which are completely defined by the error system responses. We show that imposing a finite impulse response~(FIR) constraint on the error system responses, ensures containment of all error trajectories over an infinite horizon in a finite sequence of \SL-DRS.

Consider the state and input trajectories parameterized by the affine system responses, i.e.,~$\mathbf{x} = \bm{\phi}_\mathbf{z} + \bm{\Phi}_\mathbf{e}^\mathbf{0} x_0 + \bm{\Phi}_\mathbf{e}^\mathbf{w} \mathbf{w}$ and $\mathbf{u} = \bm{\phi}_\mathbf{v} + \bm{\Phi}_\mathbf{k}^\mathbf{0} x_0 + \bm{\Phi}_\mathbf{k}^\mathbf{w} \mathbf{w}$, where~$\bm{\Phi}_\mathbf{e/k}^\mathbf{0}, \bm{\Phi}_\mathbf{e/k}^\mathbf{w}$ are partitions of the error system responses such that~$\bm{\Phi}_\mathbf{e} = \begin{bmatrix} \bm{\Phi}_\mathbf{e}^\mathbf{0} & \bm{\Phi}_\mathbf{e}^\mathbf{w} \end{bmatrix}, \, \bm{\Phi}_\mathbf{k} = \begin{bmatrix} \bm{\Phi}_\mathbf{k}^\mathbf{0} & \bm{\Phi}_\mathbf{k}^\mathbf{w} \end{bmatrix}$. Then, we separate the nominal and error trajectories like
\begin{alignat*}{2}
\mathbf{z} &= \bm{\phi}_\mathbf{z} + \bm{\Phi}_\mathbf{e}^\mathbf{0}x_0, \quad &&\mathbf{v} = \bm{\phi}_\mathbf{v} + \bm{\Phi}_\mathbf{k}^\mathbf{0}x_0, \\
\mathbf{e} &= \bm{\Phi}_\mathbf{e}^\mathbf{w}\mathbf{w}, \quad &&\:\!\mathbf{k} = \bm{\Phi}_\mathbf{k}^\mathbf{w}\mathbf{w}.
\end{alignat*}
Note that~\eqref{affine_SLS:affine-halfspace} enforces~$\phi_z^0 = 0$ and~$\Phi_e^0 = I$, therefore the error state trajectory~$\mathbf{e}$ is initialized at the origin, i.e.,~$e_0 = 0$.
The trajectories~$\mathbf{e}, \, \mathbf{k}$ parameterize the error states and inputs over the horizon~$N$, however we are interested in analyzing the error dynamics over an infinite horizon. To achieve this, we introduce a finite impulse response~(FIR) constraint on the error system responses~$\bm{\Phi}_\mathbf{e}, \, \bm{\Phi}_\mathbf{k}$.
\begin{definition}[FIR Constraint]
We call the error system responses~$\bm{\Phi}_\mathbf{e}, \, \bm{\Phi}_\mathbf{k}$ FIR-constrained if
\begin{equation}\label{eq:FIR-constraint}
\Phi_e^N = \bm{0}, \quad \Phi_k^N = \bm{0},
\end{equation}
where~$N$ is the FIR horizon.
\end{definition}
Assuming that~$\bm{\Phi}_\mathbf{e}, \, \bm{\Phi}_\mathbf{k}$ are FIR-constrained, we can rewrite the error dynamics~\eqref{eq:error-dynamics} as an infinite sequence starting in~$e_0 = 0$, with respect to the error system responses and in convolutional form:
\begin{subequations}\label{SLDRS:error-system}
\begin{align}
        e_i &= \left\lbrace\begin{matrix} \sum_{j=0}^{i-1} \Phi_e^jw_{i-1-j}, & \text{if } i = 0, \dots, N, \\[0.3em] \sum_{j=0}^{N} \Phi_e^jw_{i-1-j}, & \text{if } i > N, \end{matrix}\right. \\
        k_i &= \left\lbrace\begin{matrix} \sum_{j=0}^{i-1} \Phi_k^jw_{i-1-j}, & \text{if } i = 0, \dots, N, \\[0.3em] \sum_{j=0}^{N} \Phi_k^jw_{i-1-j}, & \text{if } i > N. \end{matrix}\right.
\end{align}
\end{subequations}
Given these error trajectories, we first define the \SL-DRS as the outer bounding sets containing all possible~$e_i, \, k_i$ for~$i=0,\dots,N$, before showing that the infinite trajectories~\eqref{SLDRS:error-system} are always contained in a finite sequence of \SL-DRS.
\begin{definition}[\SL-DRS]\label{def:SL-DRS}
        Given the FIR-constrained error system responses~$\bm{\Phi}_\mathbf{e}, \, \bm{\Phi}_\mathbf{k}$, the system level disturbance reachable sets~(\SL-DRS) are defined as the two sequences~$\{\mathcal{F}_{e,0}, \dots, \mathcal{F}_{e,N}\}$ and~$\{\mathcal{F}_{k,0}, \dots, \mathcal{F}_{k,N}\}$:
        \begin{subequations}\label{eq:SL-DRS}
        \begin{align}
        \mathcal{F}_{e,i}\left(\bm{\Phi}_\mathbf{e}\right) &\coloneqq \bigoplus_{j=0}^{i-1} \Phi_{e}^j \mathcal{W}, \qquad i = 0, \dots, N, \\
        \mathcal{F}_{k,i}\left(\bm{\Phi}_\mathbf{k}\right) &\coloneqq \bigoplus_{j=0}^{i-1} \Phi_{k}^j \mathcal{W}, \qquad i = 0, \dots, N,
        \end{align}
        \end{subequations}
        with $\mathcal{F}_{e,0}\left(\bm{\Phi}_\mathbf{e}\right) = \{0\}$ and $\mathcal{F}_{k,0}\left(\bm{\Phi}_\mathbf{k}\right) =  \{0\}$ by convention.
\end{definition}
The FIR constraint~\eqref{eq:FIR-constraint} forces the error system, described by the error system responses~$\bm{\Phi}_\mathbf{e}, \, \bm{\Phi}_\mathbf{k}$, to reach a steady-state within $N\!+\!1$ time steps. This intuitively means that the error trajectories~\eqref{SLDRS:error-system} will stop diverging when the error system reaches the steady-state. The following theorem formalizes this intuition.
\begin{theorem}\label{theorem:SL-invariance}
	If the error system responses~$\bm{\Phi}_\mathbf{e},\,\bm{\Phi}_\mathbf{k}$ fulfill the FIR constraint~\eqref{eq:FIR-constraint} for system~\eqref{eq:dynamics}, then the control law
	\begin{equation}\label{eq:RPI-law}
	\pi\!\left(\mathbf{w}^N_i\right) = \sum_{j=0}^{N-1}\Phi_k^{j}w_{i-1-j}, \quad \forall i \geq N,
	\end{equation}
	where $\mathbf{w}^N_i = (w_{i-N}, \dots, w_{i-1})$ contains the past $N$ disturbances, keeps the error~$e_i, \, i \geq N$ in the set
	\begin{equation}
	\mathcal{F}_{e,N}\left(\bm{\Phi}_\mathbf{e}\right) = \bigoplus_{i=0}^{N-1} \Phi_e^i \mathcal{W}.
	\end{equation}
	Therefore, all possible error trajectories~\eqref{SLDRS:error-system} starting in~$e_0 = 0$, lie in the sequence of \SL-DRS~\eqref{eq:SL-DRS}.
\end{theorem}
\begin{proof}
Note that~$e_i \in \mathcal{F}_{e,i}\left(\bm{\Phi}_\mathbf{e}\right), \ k_i \in \mathcal{F}_{k,i}\left(\bm{\Phi}_\mathbf{k}\right)$ for ${i \in [0,N]}$ is a direct consequence of Definition~\ref{def:SL-DRS}. Hence to prove the theorem, we only need to show that
\begin{equation}\label{proof:assumption}
e_i \in \mathcal{F}_{e,N}\left(\bm{\Phi}_\mathbf{e}\right) \implies e_{i+1} \in \mathcal{F}_{e,N}\left(\bm{\Phi}_\mathbf{e}\right), \quad \forall i \geq N,
\end{equation}
under control law~\eqref{eq:RPI-law}, which fulfills~$\pi\!\left(\mathbf{w}^N_i\right) \in \mathcal{F}_{k,N}$ by definition. We observe that, due to the FIR constraint, any error state~$e_i \in \mathcal{F}_{e,N}\left(\bm{\Phi}_\mathbf{e}\right), \, i \geq N,$ is defined in terms of the past~$N$ disturbances, i.e.,
\begin{equation*}
e_i = \sum_{j=0}^{N-1}\Phi_e^{j}w_{i-1-j},
\end{equation*}
and that these past disturbances are known for~$i \geq N$. Using these observations, we can prove the theorem by propagating the dynamics:
\begin{align*}
e_{i+1} &= Ae_{i} + B\pi\!\left(\mathbf{w}^N_i\right) + w_{i} \\
&= A\sum_{j=0}^{N-1}\Phi_e^{j}w_{i-1-j} + B\sum_{j=0}^{N-1}\Phi_k^{j}w_{i-1-j} + w_{i} \\
&= \sum_{j=0}^{N-1}\left(A\Phi_e^{j} + B\Phi_k^{j}\right)w_{i-1-j} + w_{i} \\
&= \sum_{j=0}^{N-1}\Phi_e^{j+1}w_{i-1-j} + w_{i},
\end{align*}
where we used~\eqref{affine_SLS:affine-halfspace} for~$\Phi_e^{j}$ and~$\Phi_k^{j}$ in the last equality. Note, that~$\Phi_e^N = \bm{0}$ due to~\eqref{eq:FIR-constraint} and~$\Phi_e^0 = I$ due to~\eqref{affine_SLS:affine-halfspace}, therefore we have
\begin{align*}
e_{i+1} &= \Phi_e^{N}w_{i-N} + \sum_{j=0}^{N-2}\Phi_e^{j+1}w_{i-1-j} + \Phi_e^{0}w_{i} \\
&= \sum_{j=0}^{N-1}\Phi_e^{j}w_{i-j}.
\end{align*}
Then, observing that~$w_{i-j} \in \mathcal{W}$ for all~$j=0, \dots, N-1$ implies~$e_{i+1} \in \mathcal{F}_{e,N}\left(\bm{\Phi}_\mathbf{e}\right)$, which concludes the proof.
\end{proof}
The \SL-DRS therefore share some similarities with the DRS and the minimal RPI set. Similar to the DRS approaching the minimal RPI set in the limit, the \SL-DRS approach the set~$\mathcal{F}_{e,N}\left(\bm{\Phi}_\mathbf{e}\right)$, both of which contain all possible infinite horizon error trajectories. However, instead of approaching the set in the limit, the \SL-DRS approach it in~$N\!+\!1$ time steps.

\section{TUBE-BASED MPC USING \SL-DRS}\label{sec:SLTMPC}
The system level \tube-MPC~(SLTMPC) method introduced in~\cite{Sieber2021} requires an RPI terminal set, which is in general more difficult to compute than a PI terminal set. In the following, we formulate a \tube-based MPC method using the \SL-DRS as tubes, which is similar to the SLTMPC method but only requires a PI terminal set. Additionally, the proposed method allows concurrent optimization of the nominal trajectory and the tubes. Finally, we show that the \SL-DRS tubes can also be computed offline and compare the proposed MPC method to \ct-MPC and \RPI-tube-MPC.

\subsection{FIR-Constrained System Level Tube-MPC}
Before formulating the FIR-constrained SLTMPC method, we make a few standard MPC assumptions.
\begin{assumption}
The stage cost function~$l(x,u)$ is convex, uniformly continuous for all $x \in \mathcal{X}, \,u \in \mathcal{U}$, and there exists a $\mathcal{K}_\infty$-function~$\alpha(\cdot)$ satisfying~$l(x,u) \geq \alpha(\|x\|), \, \forall x \in \mathcal{X}, \,u \in \mathcal{U}$. The terminal set~$\mathcal{X}_f$ is PI for system~\eqref{eq:dynamics} governed by the control law~$\kappa_f(x): \mathbb{R}^n \to \mathbb{R}^m$ and it holds that~$\mathcal{X}_f \subseteq \mathcal{X} \ominus \mathcal{F}_{e,N}\left(\bm{\Phi}_\mathbf{e}\right)$ and~$\kappa_f(x) \in \mathcal{U} \ominus \mathcal{F}_{k,N}\left(\bm{\Phi}_\mathbf{k}\right), \ \forall x \in \mathcal{X}_f$. The terminal cost~$l_f(x)$ is convex, uniformly continuous, there exists a $\mathcal{K}_\infty$-function $\beta(\cdot)$ satisfying~$l_f(x) \leq \beta(\| x \|), \, \forall x \in \mathcal{X}_f$, and it holds that~$l_f(x^+) - l_f(x) \leq -l(x,\kappa_f(x))$ for the nominal system~$x^{+} = Ax + B\kappa_f(x), \, \forall x \in \mathcal{X}_f$. \label{assump:SLTMPC}
\end{assumption}
Then, we can formulate the \emph{FIR-constrained SLTMPC} as
\begin{subequations}\label{SLTMPC}
\begin{alignat}{2}
	\min_{\bm{\Phi}_\mathbf{x},\bm{\Phi}_\mathbf{u}} \quad & l_f\!\left(\phi_z^N\right) + \sum_{i=0}^{N-1} l\!\left(\phi_z^i + \Phi_e^i x_0, \,\phi_v^i + \Phi_k^i x_0\right) \label{SLTMPC:cost}\\
	\textrm{s.t. } \:\; & \begin{bmatrix} I - \mathcal{Z}\mathcal{A} & -\mathcal{Z}\mathcal{B}\end{bmatrix} \begin{bmatrix} \bm{\Phi}_\mathbf{x} \\ \bm{\Phi}_\mathbf{u} \end{bmatrix} = I, \label{SLTMPC:dyn}\\
	& \Phi_e^N = \bm{0}, \ \Phi_k^N = \bm{0}, \label{SLTMPC:FIR}\\
	& \phi_z^i + \Phi_e^i x_0 \in \mathcal{X} \ominus \mathcal{F}_{e,i}\left(\bm{\Phi}_\mathbf{e}\right),  &&\hspace*{-1.5cm} i = 0, \dots, N\!-\!1, \label{SLTMPC:state_constraints} \\
        & \phi_v^i + \Phi_k^i x_0 \in \mathcal{U} \ominus \mathcal{F}_{k,i}\left(\bm{\Phi}_\mathbf{k}\right),  &&\hspace*{-1.5cm} i = 0, \dots, N\!-\!1, \label{SLTMPC:input_constraints} \\
        & \phi_z^N \in \mathcal{X}_f, \label{SLTMPC:terminal_constraint} \\
        & x_0 = x(t),
\end{alignat}
\end{subequations}
where the stage cost function~$l(\cdot,\cdot)$, terminal cost function~$l_f(\cdot)$, and terminal set~$\mathcal{X}_f$ satisfy Assumption~\ref{assump:SLTMPC}. The resulting MPC control law is given by
\begin{equation}\label{SLTMPC:control-law}
\kappa(x(t)) = \phi_v^{0,*} + \Phi_k^{0,*} x(t),
\end{equation}
where~$\phi_v^{0,*}, \Phi_k^{0,*}$ are optimizers of~\eqref{SLTMPC}. In order to write the constraints~\eqref{SLTMPC:state_constraints} -~\eqref{SLTMPC:terminal_constraint} in closed-form, we rely on a standard technique from disturbance feedback MPC~\cite{Goulart2006}, which removes the explicit dependence on the disturbance trajectory by employing Lagrangian dual variables. This corresponds to an implicit representation of the \SL-DRS used for constraint tightening and enables the concurrent optimization of the \SL-DRS and the nominal trajectories~$\mathbf{z}, \mathbf{v}$.

Note that the choice of the PI terminal set~$\mathcal{X}_f$ needs additional discussion, since it is required to be a subset of the tightened state constraints and the tightening is not known prior to solving~\eqref{SLTMPC}. Therefore, we need to either choose~$\mathcal{X}_f$ such that it stays PI under any tightening or define~$\mathcal{X}_f$ implicitly, which also enables optimization over~$\mathcal{X}_f$ when solving~\eqref{SLTMPC}. While this restricts the choice of~$\mathcal{X}_f$ compared to \RPI-tube-MPC, it is easy to find an~$\mathcal{X}_f$ which fulfills one of the above criteria; we discuss three possible options below. The first option is to choose the set of steady-states for system~\eqref{eq:dynamics}, since this set remains PI under any tightening, assuming at least one steady-state is contained in the tightened terminal set. The second option is to use a scaled version of any PI set, i.e.,~$\mathcal{X}_f = \lambda \mathcal{S}, \lambda \geq 0$ where~$\mathcal{S}$ is PI, and also optimize over the scaling~$\lambda$ in~$\eqref{SLTMPC}$ such that~$\mathcal{X}_f \subseteq \mathcal{X} \ominus \mathcal{F}_{e,N}\left(\bm{\Phi}_\mathbf{e}\right)$ and~$\kappa_f(x) \in \mathcal{U} \ominus \mathcal{F}_{k,N}\left(\bm{\Phi}_\mathbf{k}\right), \ \forall x \in \mathcal{X}_f$ are fulfilled. The third option is to define the terminal set implicitly through the feasible set of a nominal MPC regulator. This can be achieved by choosing the MPC horizon longer than the FIR horizon, i.e.,~$N_\textup{MPC} > N$ and adding the constraints~$\phi_z^{N_\textup{MPC}} = 0, \, \phi_v^{N_\textup{MPC}}=0$ to~\eqref{SLTMPC}. Note that these options are examples and there exist other terminal sets fulfilling the necessary criteria.

Next, we show that the application of~\eqref{SLTMPC:control-law} to system~\eqref{eq:dynamics} results in a closed-loop system, which is input-to-state stable~(ISS)~\cite[Definition~3]{Limon2009} in the RPI set spanned by all initial states~$x(t)$ for which~\eqref{SLTMPC} is feasible.
\begin{proposition}\label{prop:SLTMPC}
Given that Assumption~\ref{assump:SLTMPC} holds, system~\eqref{eq:dynamics} subject to constraints~\eqref{eq:constraints} and controlled by control law~\eqref{SLTMPC:control-law}, is ISS in $\mathcal{X}_\textup{feas}$, which is the set of all initial states~$x(t)$ for which~\eqref{SLTMPC} is feasible.
\end{proposition}
\begin{proof}
To prove the proposition, we first show that~\eqref{SLTMPC} is recursively feasible, before showing that the value function of~\eqref{SLTMPC} is an ISS Lyapunov function in~$\mathcal{X}_\textup{feas}$~(Definition~\ref{def:ISS-Lyap}). Subsequently, we use~\cite[Theorem~3]{Limon2009} to conclude that the uncertain system under control law~\eqref{SLTMPC:control-law} is ISS in~$\mathcal{X}_\textup{feas}$.
\newline\textit{Recursive Feasibility:} Consider a state~$x(t) \in \mathcal{X}_\textup{feas}$ and denote by~$\mathbf{z}^*(x(t)), \, \mathbf{v}^*(x(t))$ the optimal solution of~\eqref{SLTMPC} with initial condition~$x(t)$. Let~$x^+ = x(t\!+\!1)$ be the actual state at the next sampling time and define the sequences~$\hat{\mathbf{z}}(x^+), \hat{\mathbf{v}}(x^+)$ as the shifted previous solutions applied at this state, i.e.,
\begin{align*}
\hat{\mathbf{z}}(x^+) &\coloneqq (z_1^*, \dots, z_{N}^*, Az_{N}^*\!+\!B\kappa_f(z_N^*)), \\
\hat{\mathbf{v}}(x^+) &\coloneqq (v_1^*, \dots, v_{N-1}^*, \kappa_f(z_N^*)).
\end{align*}
Then, we construct the candidate sequences for the initial state~$x^+$ as~$\bar{\mathbf{z}}(x^+) \coloneqq (\bar{z}_0, \dots, \bar{z}_N), \, \bar{\mathbf{v}}(x^+) \coloneqq (\bar{v}_0, \dots, \bar{v}_{N-1})$, where
\begin{align*}
\bar{z}_j &= \hat{z}_j + \Phi_e^{j,*}w(t), \quad j=0,\dots,N, \\
\bar{v}_j &= \hat{v}_j + \Phi_k^{j,*}w(t), \quad j=0,\dots,N\!-\!1,
\end{align*}
and~$\Phi_e^{j,*}, \, \Phi_k^{j,*}$ are the error system responses from the previous solution. We verify that~$\bar{z}_j \in \mathcal{X} \ominus \mathcal{F}_{e,j}\left(\bm{\Phi}^*_\mathbf{e}\right)$ for~${j=0,\dots,N\!-\!1}$, since~$\left\lbrace \hat{z}_j \right\rbrace \oplus \mathcal{F}_{e,j+1}\left(\bm{\Phi}^*_\mathbf{e}\right) \subseteq \mathcal{X} $ and~$\left\lbrace \Phi_e^{j,*}w(t) \right\rbrace \oplus \mathcal{F}_{e,j}\left(\bm{\Phi}^*_\mathbf{e}\right) \subseteq \mathcal{F}_{e,j+1}\left(\bm{\Phi}^*_\mathbf{e}\right)$ due to~\eqref{eq:SL-DRS}. For~$j=N$, we have $\bar{z}_N = \hat{z}_N$ due to the FIR constraint~\eqref{SLTMPC:FIR}, which allows us to verify that $\bar{z}_N \in \mathcal{X}_f$ since~$\mathcal{X}_f$ is PI. Analogously, we can show that~$\bar{v}_j \in \mathcal{U} \ominus \mathcal{F}_{k,j}\left(\bm{\Phi}_\mathbf{k}\right)$ for~$j=0,\dots,N\!-\!2$. For~${j=N\!-\!1}$, $\bar{v}_{N-1} \in \mathcal{U} \ominus \mathcal{F}_{k,N-1}\left(\bm{\Phi}_\mathbf{k}\right)$ since~$\hat{v}_{N-1} = \kappa_f(z_N^*) \in \mathcal{U} \ominus \mathcal{F}_{k,N}\left(\bm{\Phi}_\mathbf{k}\right)$ and~$\{ \Phi_k^{N-1,*}w(t) \} \oplus \mathcal{F}_{k,N-1}\left(\bm{\Phi}^*_\mathbf{k}\right) \subseteq \mathcal{F}_{k,N}\left(\bm{\Phi}^*_\mathbf{k}\right)$ due to~\eqref{eq:SL-DRS}. Therefore,~$\bar{\mathbf{z}}(x^+), \bar{\mathbf{v}}(x^+)$ are feasible solutions of~\eqref{SLTMPC} for~$x(t+1)$ and recursive feasibility is proven.
\newline\textit{ISS Lyapunov Function:}
Denote~\eqref{SLTMPC:cost} as~$V_N(x(t); \mathbf{z}, \mathbf{v})$ and the value function of~\eqref{SLTMPC} as~$V_N^*(x(t))$, i.e.,~$V_N(\cdot)$ evaluated at the optimal solution. Then, it follows that
\begin{align*}
V_N^*(x(t\!+\!1)) - V_N^*(x(t)) &\leq V_N(x(t\!+\!1); \bar{\mathbf{z}}, \bar{\mathbf{v}}) - V_N^*(x(t)) \\
&\mathrel{\overset{\makebox[0pt]{\mbox{\normalfont\scriptsize (a)}}}{=}} l_f(\bar{z}_N) \!-\! l_f(z_N^*) \!+\! l(\hat{z}_{N\!-\!1},\hat{v}_{N\!-\!1}) \\
&\quad- l(z_0^*,v_0^*)  \\
&\quad+ \sum_{j=0}^{N-1} \left( l(\bar{z}_{j}, \bar{v}_{j}) - l(\hat{z}_j, \hat{v}_j)\right) \\
&\mathrel{\overset{\makebox[0pt]{\mbox{\normalfont\scriptsize (b)}}}{\leq}} - l(z_0^*,v_0^*) \\
&\quad+ \sum_{j=0}^{N-1} \left( l(\bar{z}_{j}, \bar{v}_{j}) - l(\hat{z}_j, \hat{v}_j)\right)
\end{align*}
where, $\bar{\mathbf{z}}, \bar{\mathbf{v}}$ are the candidate sequences. In~(a), we summarized the stage costs under one sum by adding the term~$l(\hat{z}_{N\!-\!1},\hat{v}_{N\!-\!1}) - l(\hat{z}_{N\!-\!1},\hat{v}_{N\!-\!1})$ on the right hand side and by using the fact that~$z_j^* = \hat{z}_{j-1}, v_j^* = \hat{v}_{j-1}$ by definition. For~(b) we used Assumption~\ref{assump:SLTMPC}, in particular~$l_f(Az_N^* + B\kappa_f(z_N^*)) - l_f(z_N^*) \leq -l(z_N^*,\kappa_f(z_N^*))$, since~$z_N^* \in \mathcal{X}_f$. Next, we deduct that~$-l(z_0^*,v_0^*) \leq -\alpha(\|x(t)\|)$ due to Assumption~\ref{assump:SLTMPC} and~$z_0^* = x(t)$. Note that, due to~\cite[Lemma~1]{Limon2009} and uniform continuity of the stage cost (Assumption~\ref{assump:SLTMPC}), there exist some~$\mathcal{K}_\infty$-functions~$\sigma_x, \, \sigma_u$ such that
\begin{align*}
l(\bar{z}_{j},\! \bar{v}_{j}) \!-\! l(\hat{z}_j, \!\hat{v}_j) \!&\leq\! \bigl\vert l(\hat{z}_j \!+\! \Phi_e^{j,*}\!\!w(t), \hat{v}_j \!+\! \Phi_k^{j,*}\!\!w(t)) \!-\! l(\hat{z}_j, \!\hat{v}_j) \bigr\vert \\
\!&\leq \sigma_x(\|\Phi_e^{j,*}\! w(t)\|) \!+\! \sigma_u(\|\Phi_k^{j,*}\! w(t)\|),
\end{align*}
for all~$j = 0, \dots, N\!-\!1$. Then, there exists a $\mathcal{K}$-function~$\gamma$ such that $V_N^*(x(t\!+\!1)) - V_N^*(x(t)) \leq \!-\alpha(\|x(t)\|) + \gamma(\|w(t)\|)$, which proves that~$V^*_N(x)$ fulfills~\eqref{ISS-L:decrease}.
Since~$V^*_N(x) \geq l(x,\kappa(x)), \, \forall x \in \mathcal{X}_\textup{feas}$, the optimal value function is lower bounded by the $\mathcal{K}_\infty$-function~$\alpha$ by assumption, i.e.,~$V^*_N(x) \geq \alpha(\|x\|), \, \forall x \in \mathcal{X}_\textup{feas}$. Using the facts that~$l_f(x)$ is upper bounded, the state and input constraints are compact, and the optimal value function~$V^*_N(x)$ is monotone, we deduct that~$V^*_N(x)$ is upper bounded by a~$\mathcal{K}_\infty$-function in~$\mathcal{X}_\textup{feas}$ due to~\cite[Propositions 2.15 - 2.16]{Rawlings2009}. This proves that~$V^*_N(x)$ also fulfills~\eqref{ISS-L:bounds} and that~$V^*_N(x)$ is an ISS Lyapunov function in~$\mathcal{X}_\textup{feas}$.
\end{proof}

\subsection{Computing the \SL-DRS Offline}
Instead of computing the \SL-DRS and nominal trajectories~$\mathbf{z}, \mathbf{v}$ concurrently, we can also compute the \SL-DRS offline and only optimize the nominal trajectories online. The offline computation is formulated as
\begin{subequations}\label{SLTMPC:tube-comp}
\begin{align}
	\min_{\bar{\bm{\Phi}}_\mathbf{e},\bar{\bm{\Phi}}_\mathbf{k}} \quad & L(\bar{\bm{\Phi}}_\mathbf{e}, \bar{\bm{\Phi}}_\mathbf{k}) \label{offTubeSLTMPC:cost}\\
	\textrm{s.t. } \:\; & \begin{bmatrix} I - \mathcal{Z}\mathcal{A} & -\mathcal{Z}\mathcal{B}\end{bmatrix} \begin{bmatrix} \bar{\bm{\Phi}}_\mathbf{e} \\ \bar{\bm{\Phi}}_\mathbf{k} \end{bmatrix} = I, \label{offTubeSLTMPC:dyn}\\
	& \bar{\Phi}_e^N = \bm{0}, \ \bar{\Phi}_k^N = \bm{0},\\
        & \mathcal{F}_{e,N}\left(\bar{\bm{\Phi}}_\mathbf{e}\right) \subseteq \mathcal{X}, \quad \mathcal{F}_{k,N}\left(\bar{\bm{\Phi}}_\mathbf{k}\right) \subseteq \mathcal{U}, \label{offTubeSLTMPC:terminal_constraint}
\end{align}
\end{subequations}
where~$L(\cdot,\cdot)$ is an appropriate cost function. The choice of~$L(\cdot,\cdot)$ heavily influences the shape and characteristics of the resulting \SL-DRS. For example~$L(\bar{\bm{\Phi}}_\mathbf{e}, \bar{\bm{\Phi}}_\mathbf{k}) = \|\,\mathcal{C}\begin{bmatrix} \bar{\bm{\Phi}}_\mathbf{e}^\top & \bar{\bm{\Phi}}_\mathbf{k}^\top \end{bmatrix}^\top\!\|_{\infty \to \infty}$, where~$\|\cdot\|_{\infty\to\infty}$ is the induced maximum norm and~$\mathcal{C}$ an appropriate weighting matrix, optimizes the worst case gain with respect to the infinity norm from disturbance signal~$\mathbf{w}$ to error signals~$\mathbf{e}, \mathbf{k}$.

Since we compute the \SL-DRS offline, we also compute the tightenings offline and can thus employ a PI terminal set~$\mathcal{X}_f$ with no additional restrictions. Therefore, we define the MPC problem using the precomputed \SL-DRS similar to the \ct-MPC and \RPI-tube-MPC problems:
\begin{subequations}\label{SLTMPC:offline}
\begin{alignat}{2}
	\min_{\mathbf{z},\mathbf{v}} \quad & l_f(z_N) + \sum_{i=0}^{N-1}l(z_i,v_i) \label{offSLTMPC:cost}\\
	\textrm{s.t. } \:\; & z_{i+1} = Az_i + Bv_i,  &&i = 0, \dots, N\!-\!1, \label{offSLTMPC:dyn}\\
	& z_i \in \mathcal{X} \ominus \mathcal{F}_{e,i}\left(\bar{\bm{\Phi}}_\mathbf{e}\right),  &&i = 0, \dots, N\!-\!1, \label{offSLTMPC:state_constraints}\\
        & v_i \in \mathcal{U} \ominus \mathcal{F}_{k,i}\left(\bar{\bm{\Phi}}_\mathbf{k}\right),  &&i = 0, \dots, N\!-\!1, \label{offSLTMPC:input_constraints} \\
        & z_N \in \mathcal{X}_f, \label{offSLTMPC:terminal_constraint} \\
        & z_0 = x(t),
\end{alignat}
\end{subequations}
where~$\mathcal{F}_{e,i}\left(\bar{\bm{\Phi}}_\mathbf{e}\right)$ and~$\mathcal{F}_{k,i}\left(\bar{\bm{\Phi}}_\mathbf{k}\right)$ are precomputed using~\eqref{SLTMPC:tube-comp} and the stage cost function~$l(\cdot,\cdot)$, terminal cost function~$l_f(\cdot)$, and terminal set~$\mathcal{X}_f$ satisfy Assumption~\ref{assump:SLTMPC}. Similar to the online version of the FIR-constrained SLTMPC method in~\eqref{SLTMPC}, problem~\eqref{SLTMPC:offline} produces a control law which ensures ISS of the closed-loop system.
\begin{remark}
The two FIR-constrained SLTMPC methods~\eqref{SLTMPC} and~\eqref{SLTMPC:offline} can also handle min-max cost functions such as $H_\infty$ or $L_1$, since the affine SLP allows the implicit reformulation of the inner maximization problem, see~\cite[Section~2.2]{Anderson2019} for more details. However, using a min-max cost function in~\eqref{SLTMPC} requires a separate discussion of ISS of the resulting FIR-constrained SLTMPC problem similar to~\cite[Section~5.3]{Limon2009}.
\end{remark}

\subsection{Comparison to \ct-MPC and \RPI-tube-MPC}
Since the FIR-constrained SLTMPC method concurrently optimizes the \SL-DRS and the nominal trajectories, its computational complexity is higher than the computational complexity of \ct-MPC and \RPI-tube-MPC. However, if the \SL-DRS are computed offline, then the online computations are reduced to optimizing the nominal trajectories and thus the computational complexity is similar to that of \ct-MPC and \RPI-tube-MPC. As a consequence of Theorem~\ref{theorem:SL-invariance}, the \SL-DRS combine the growing set characteristic of DRS with the ability to contain the infinite horizon evolution. Therefore, the MPC methods~\eqref{SLTMPC} and~\eqref{SLTMPC:offline} combine the characteristics of \ct-MPC and \RPI-tube-MPC and thereby alleviate their individual drawbacks, namely the RPI terminal set for \ct-MPC and that full control authority is not available at the initial state for \RPI-tube-MPC.

\section{NUMERICAL RESULTS}\label{sec:numerical_section}
In the following, we compare the proposed FIR-constrained SLTMPC method~\eqref{SLTMPC} to the SLTMPC method proposed in~\cite{Sieber2021}. Subsequently, we show that precomputing the \SL-DRS and using~\eqref{SLTMPC:offline} outperforms \ct-MPC and \RPI-tube-MPC. Finally, we show how the choice of particular cost functions in~\eqref{SLTMPC:tube-comp} affects the resulting \SL-DRS. All examples are implemented in Python using CVXPY~\cite{cvxpy} and are solved using MOSEK~\cite{mosek}. The examples were run on a machine equipped with an Intel~i7-8665U~(\unit[1.9]{GHz}) CPU and \unit[16]{GB} of RAM. In all experiments, we consider the uncertain LTI system~\eqref{eq:dynamics} with discrete-time dynamic matrices
\begin{equation*}
A = \begin{bmatrix} 1.05 & 0.15 \\ 0 & 1 \end{bmatrix}, \quad B = \begin{bmatrix} 0.5 \\ 0.5 \end{bmatrix},
\end{equation*}
subject to the polytopic constraints
\begin{equation*}
\begin{bmatrix} \scalebox{0.7}[1.0]{\( - \)}1 \\ \scalebox{0.7}[1.0]{\( - \)}1.5 \end{bmatrix} \!\!\leq\!\! \begin{bmatrix} x_1 \\ x_2 \end{bmatrix} \!\!\leq\!\! \begin{bmatrix} 0.5 \\ 1.5 \end{bmatrix}\!, \scalebox{0.7}[1.0]{\( - \)}0.5 \!\leq\! u \!\leq\! 0.5, \begin{bmatrix} \scalebox{0.7}[1.0]{\( - \)}\theta \\ \scalebox{0.7}[1.0]{\( - \)}0.1 \end{bmatrix} \!\!\leq\!\! \begin{bmatrix} w_1 \\ w_2 \end{bmatrix} \!\!\leq\!\! \begin{bmatrix} \theta \\ 0.1 \end{bmatrix}\!,
\end{equation*}
cost function~$l(z,v) = z^\top\! Qz + v^\top\! Rv$ with $Q=100\cdot I$ and $R=10$, horizon $N=10$, and parameter $\theta$. In all experiments, we use the same constant feedback controller~$K$, which is computed such that~$K$ minimizes the constraint tightening, using the method proposed in~\cite[Section~7.2]{Limon2010}.
\begin{figure}
\centering
\includegraphics[width=0.96\linewidth]{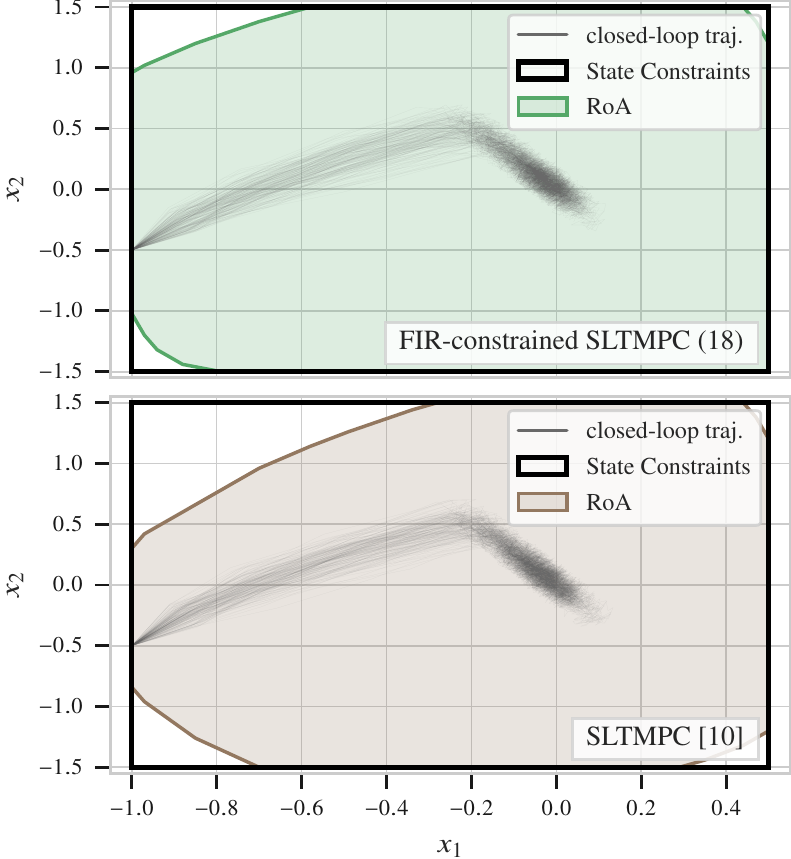}
\caption{RoA and 200 closed-loop trajectories for FIR-constrained SLTMPC~(top) and SLTMPC~(bottom), with $x_0 = [-1.0, -0.5]^\top\!,$ and $\theta=0.04$.}
\label{fig:SLTMPC-vs-FIR}
\end{figure}

\subsection{Comparison of FIR-constrained SLTMPC~\eqref{SLTMPC} to SLTMPC~\cite{Sieber2021}}
To assess the effectiveness of the proposed FIR-constrained SLTMPC~\eqref{SLTMPC}, we compare it to SLTMPC~\cite{Sieber2021}. As the terminal set for SLTMPC we choose the maximal RPI set for system~\eqref{eq:dynamics} subject to constraints~\eqref{eq:constraints} under the constant feedback controller~$K$. For~\eqref{SLTMPC}, we choose~$\mathcal{X}_f = \lambda \mathcal{S}$, where~$\mathcal{S}$ is the maximal PI set for system~\eqref{eq:dynamics} subject to~\eqref{eq:constraints} under the same controller~$K$ and optimize~$\lambda$ in~\eqref{SLTMPC}. Figure~\ref{fig:SLTMPC-vs-FIR} shows the regions of attraction~(RoA) and 200 closed-loop trajectories with different noise realizations for both methods and parameter choice~$\theta = 0.04$. Both methods compute similar closed-loop trajectories with mean closed-loop costs of \unit[523.6]{} for~\eqref{SLTMPC} and \unit[531.8]{} for SLTMPC. However, the RoA for FIR-constrained SLTMPC is larger than for SLTMPC. In this experiment, both methods achieved similar average solve times of~\unit[30-35]{ms}.
\begin{figure}
\centering
\includegraphics[width=0.96\linewidth]{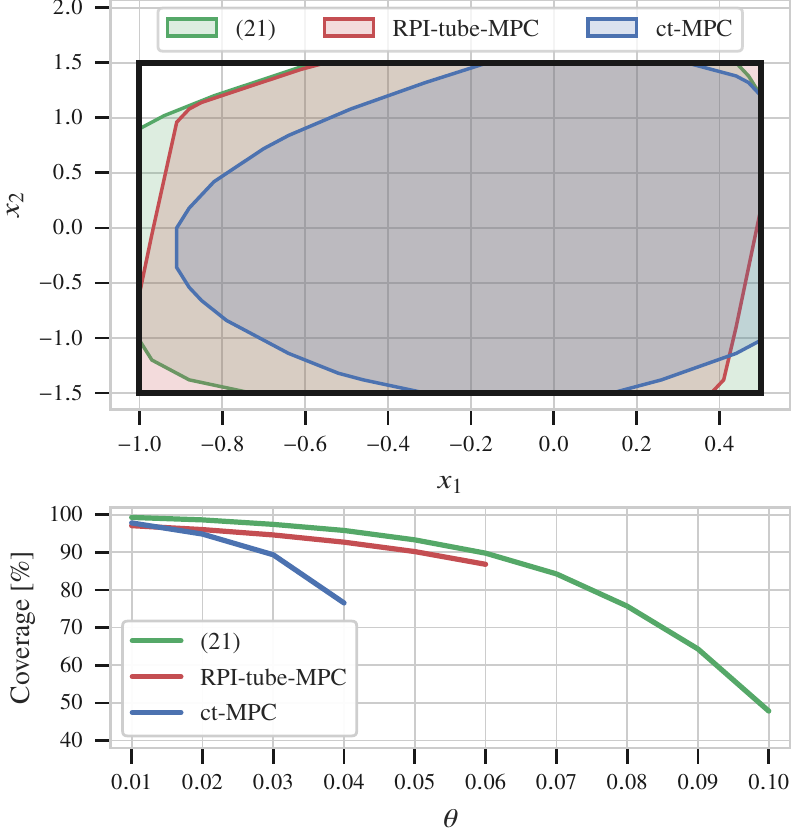}
\caption{RoA of FIR-constrained SLTMPC, RPI-tube-MPC, and ct-MPC for~${\theta = 0.04}$~(top) and approximate RoA coverage with respect to the state constraints in percent, as a function of parameter~$\theta$~(bottom).}\label{fig:RoA-coverage}
\end{figure}

\subsection{Comparison of~\eqref{SLTMPC:offline} to RPI-tube-MPC and ct-MPC}
To illustrate the benefits of using varying tube controllers instead of a single fixed one, we compare the FIR-constrained SLTMPC method with precomputed \SL-DRS~\eqref{SLTMPC:offline} to \RPI-tube-MPC and \ct-MPC, which both use~$K$ as their tube controller. The tubes for~\eqref{SLTMPC:offline} are computed using~\eqref{SLTMPC:tube-comp} and a cost function which minimizes the state and input tightenings, i.e.,
\begin{equation}\label{eq:min-tight}
L(\bm{\Phi}_\mathbf{e}, \bm{\Phi}_\mathbf{k}) = \Vert\bm{\Lambda}_\mathbf{e} \mathbf{h}_\mathbf{w}\Vert_\infty + \Vert\bm{\Lambda}_\mathbf{k} \mathbf{h}_\mathbf{w}\Vert_\infty,
\end{equation}
where~$\mathbf{h}_\mathbf{w} = (h_w, \dots, h_w)$ is the $N$-times stacked disturbance vector and~$\bm{\Lambda}_\mathbf{e}, \, \bm{\Lambda}_\mathbf{k}$ are the Lagrangian dual variables. As the terminal sets, we choose the maximal PI set under controller~$K$ for FIR-constrained SLTMPC and \RPI-tube-MPC and the maximal RPI set under the same controller for \ct-MPC, respectively. Figure~\ref{fig:RoA-coverage}~(top) shows the RoAs of all three methods for~$\theta = 0.04$. FIR-constrained SLTMPC achieves a RoA similar to \RPI-tube-MPC but a significantly larger one than \ct-MPC. This is further highlighted by Figure~\ref{fig:RoA-coverage}~(bottom), which depicts the approximate coverage of the RoA, i.e., the area of the state constraint set covered by the RoA in percent. The RoA coverage of FIR-constrained SLTMPC is consistently larger than for the other two methods and it also remains feasible for significantly larger noise levels. In this setting, all three methods achieved similar closed-loop costs and average solve times of~\unit[5-6]{ms}.

\subsection{Effect of different cost functions for~\eqref{SLTMPC:tube-comp}}
The choice of cost function~\eqref{offTubeSLTMPC:cost} naturally affects the shape and size of the sets computed by~\eqref{SLTMPC:tube-comp}. Figure~\ref{fig:SL-DRS-costs} shows the sets~$\mathcal{F}_{e,N}\left(\bm{\Phi}_\mathbf{e} \right)$~(top) and~$\mathcal{F}_{k,N}\left(\bm{\Phi}_\mathbf{k} \right)$~(bottom) for four different cost functions\footnote{For the definitions of LQR,~$\mathcal{L}_1$, and~$\mathcal{H}_\infty$ costs in terms of the system responses~$\bm{\Phi}_\mathbf{e}, \, \bm{\Phi}_\mathbf{k}$, see~\cite[Section~2.2]{Anderson2019}.}. As expected, the minimal tightening cost function~\eqref{eq:min-tight} produces the smallest sets for both states and inputs. The~$\mathcal{L}_1$ and~$\mathcal{H}_\infty$ cost functions produce sets~$\mathcal{F}_{k,N}\left(\bm{\Phi}_\mathbf{k} \right)$, which are equal to the input constraints, resulting in overly conservative optimization problems~\eqref{SLTMPC:offline}. This can for example be alleviated by extending the cost function~\eqref{offTubeSLTMPC:cost} by a regularization term or by modifying the constraints~\eqref{offTubeSLTMPC:terminal_constraint}, e.g.,~$\mathcal{F}_{e,N}\left(\bm{\Phi}_\mathbf{e} \right) \subseteq \rho_x \mathcal{X}, \, \mathcal{F}_{k,N}\left(\bm{\Phi}_\mathbf{k} \right) \subseteq \rho_u \mathcal{U}$ with design parameters~$\rho_{x/u} \in [0,1]$.

\section{CONCLUSIONS}\label{sec:conclusions}
This paper has introduced system level disturbance reachable sets~(\SL-DRS) for constrained linear time-invariant systems with additive noise. We showed that imposing a finite impulse response~(FIR) constraint on the \SL-DRS guarantees that the error trajectories are contained in a finite sequence of \SL-DRS for all time steps. Using the \SL-DRS as tubes, we formulated a \tube-based MPC method dubbed FIR-constrained SLTMPC and showed that the \SL-DRS tubes can be computed both online and offline. The proposed method enables concurrent computation of the nominal trajectory and the tubes and only requires a positively invariant terminal set. Finally, we showed the benefits and behavior of both the online and offline versions on three numerical examples.

\section*{Acknowledgment}
We would like to thank Carlo Alberto Pascucci from Embotech and Simon Muntwiler from ETH Zurich for the insightful discussions.
\begin{figure}
\centering
\includegraphics[width=0.97\linewidth]{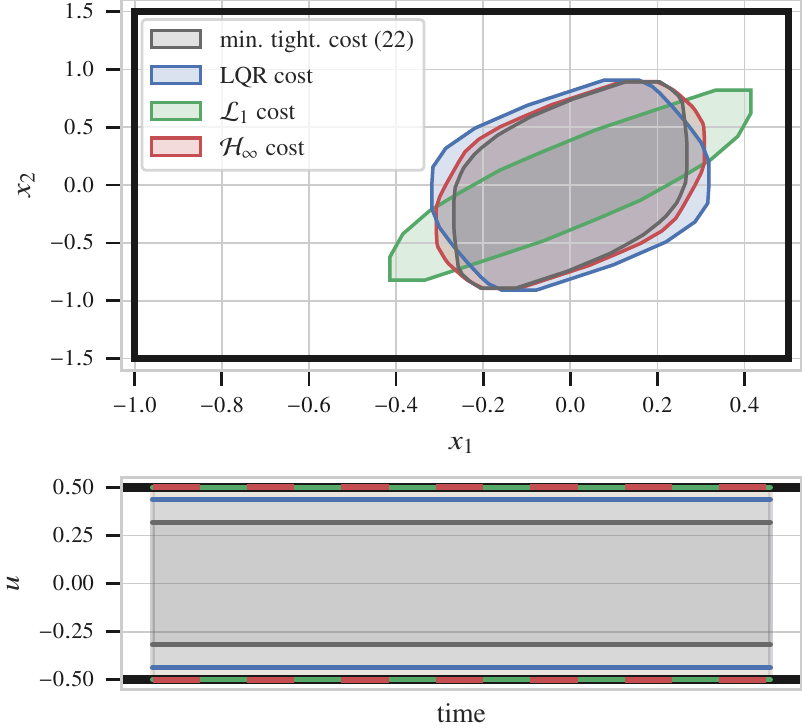}
\caption{Set~$\mathcal{F}_{e,N}\left(\bm{\Phi}_\mathbf{e} \right)$ of the state SL-DRS sequence for different cost functions~\eqref{offTubeSLTMPC:cost}~(top) and~$\mathcal{F}_{k,N}\left(\bm{\Phi}_\mathbf{k} \right)$ of the input SL-DRS sequence for the same cost functions~(bottom).}
\label{fig:SL-DRS-costs}
\end{figure}





\bibliography{bibliography.bib}

\end{document}